\documentclass{article}
\usepackage{ijcai17}

\usepackage{times}

\usepackage{amsmath}
\usepackage{amssymb}
\usepackage{amsthm}
\usepackage{pgf}
\usepackage{pgfplots}
\usepackage{pgfkeys}
\usepackage{subfig}
\usepackage{url}
\usepackage[font=small]{caption}
\usetikzlibrary{patterns}

\title{The Hard Problems Are Almost Everywhere For Random CNF-XOR Formulas	\thanks{The author list has been sorted alphabetically by last name; this should not be used to determine the extent of authors' contributions.}}
\author{Jeffrey M. Dudek\\Rice University \And Kuldeep S. Meel\\Rice Univesity \And  Moshe Y. Vardi\\Rice University
}

\begin{document}

\pgfkeys{
	/cnfxor/.is family, /cnfxor,
	default/.style = {k = {k}, n = {n}, r = {r}, s = {s}, p = {p}, cnfNum = {}, xorNum = {}},
	k/.estore in = \clauseK,
	n/.estore in = \clauseN,
	r/.estore in = \clauseR,
	s/.estore in = \clauseS,
	p/.estore in = \clauseP,
	cnfNum/.estore in = \clauseCNFOverride,
	xorNum/.estore in = \clauseXOROverride,
}

\newcommand\F[1][]{%
	\pgfkeys{/cnfxor, default, #1}%
	F_{\clauseK}(\clauseN, 
	\ifx\clauseCNFOverride\empty\relax
	\clauseR \clauseN
	\else
	\clauseCNFOverride
	\fi
	)
}

\newcommand\Q[1][]{%
	\pgfkeys{/cnfxor, default, #1}%
	Q^{\clauseP}(\clauseN, 
	\ifx\clauseXOROverride\empty\relax
	\clauseS \clauseN
	\else
	\clauseXOROverride
	\fi
	)
}

\newcommand\FQ[1][]{%
	\pgfkeys{/cnfxor, default, #1}%
	\psi^{\clauseP}_{\clauseK}(\clauseN, 
	\ifx\clauseCNFOverride\empty\relax
	\clauseR \clauseN
	\else
	\clauseCNFOverride
	\fi,
	\ifx\clauseXOROverride\empty\relax
	\clauseS \clauseN
	\else
	\clauseXOROverride
	\fi
	)
}

\renewcommand{\P}[1]{\ensuremath{\mathsf{Pr}\left[#1\right]}}
\newcommand{\Var}[1]{\ensuremath{\mathsf{Var}\left[#1\right]}}
\newcommand{\Covar}[2]{\ensuremath{\mathsf{Cov}\left[#1, #2\right]}}
\newcommand{\E}[1]{\ensuremath{\mathsf{E}\left[#1\right]}}

\renewcommand{\L}[0]{\Lambda}
\newcommand{\Lfull}[0]{\Lambda_b(k, r)}
\newcommand{\Lfullest}[0]{\Lambda_b(1/2, k, r)}

\newcommand{\SMT}{\ensuremath{\mathsf{SMT}}}
\newcommand{\SAT}{\ensuremath{\mathsf{SAT}}}

\newcommand{\seq}[2]{\{#1_#2\}_{#2 = 1}^{\infty}}
\newcommand{\ceil}[1]{\lceil #1 \rceil}
\newcommand{\floor}[1]{\lfloor #1 \rfloor}

\newcommand{\CryptoMiniSAT}{\ensuremath{\mathsf{CryptoMiniSAT}}}

\newtheorem{theorem}{Theorem}
\newtheorem{lemma}[theorem]{Lemma}
\newtheorem{corollary}[theorem]{Corollary}

\newtheorem{conjecture}{Conjecture}
\newtheorem{speculation}{Speculation}
\everymath{\textstyle}

\maketitle

\begin{abstract}
	Recent universal-hashing based approaches to sampling and counting crucially depend on the runtime performance of \SAT~solvers on
	formulas expressed as the conjunction of both CNF constraints and variable-width XOR constraints (known as CNF-XOR formulas). In this paper, we present the first study of the runtime behavior of \SAT~solvers equipped with XOR-reasoning techniques on random CNF-XOR formulas. We empirically demonstrate that a state-of-the-art \SAT~solver scales exponentially on random CNF-XOR formulas across a wide range of XOR-clause densities, peaking around the empirical phase-transition location.
	On the theoretical front, we prove that the solution space of a random CNF-XOR formula `shatters' at \emph{all} nonzero XOR-clause densities into well-separated components, similar to the behavior seen in random CNF formulas known to be difficult for many \SAT-solving algorithms.
\end{abstract}

%
\section{Introduction}\label{sec:introduction}
The Boolean-Satisfaction Problem (\SAT) is one of the most fundamental problems in 
computer science, with a wide range of applications arising from diverse areas 
such as artificial intelligence, programming languages, biology and the like~\cite{BHMW09}. While \SAT~is NP-complete, the study of the runtime behavior of \SAT~techniques is a topic of major interest in AI~\cite{LGZC15} owing to its practical usage. Of specific interest
is the behavior of \SAT~solvers on random problems~\cite{CKT91}, motivated by the connection
between the clause density (the ratio of clauses to variables) of a random \SAT~instance and algorithmic
properties of the solution space. Early experiments~\cite{MSL92,CA93,KS94} on random fixed-width CNF formulas (where each clause
contains a fixed number of literals) revealed a surprising phase-transition behavior in the  satisfiability of random formulas: the probability of satisfiability undergoes a precipitous
drop around a fixed density, the location of which depends only on the clause width (for CNF formulas with clause width 3, this occurs around a clause density of 4.26).
Moreover, the runtime of \SAT~solvers (using DPLL and related algorithms) on random CNF formulas was shown to follow an \emph{easy-hard-easy} pattern~\cite{KS94}: the runtime is low when the clause density is very low or very high and peaks near the phase-transition point.

Further analysis of the relationship between the clause density and \SAT~solver runtime revealed a more nuanced picture of the scaling behavior of \SAT~solvers on random fixed-width CNF instances: a secondary phase-transition was observed
within the satisfiable region, where the median runtime transitions from polynomial to exponential in the number of variables~\cite{CDSSV03}. Theoretical analysis of this phenomenon \cite{DMMZ08,MMZ05,ACR11} has shown that the solution
space of a random fixed-width CNF formula undergoes a dramatic `shattering'. When the clause density is small, 
almost all solutions are contained in a single connected-component (where solutions are adjacent if their
Hamming distance is 1). In this region, several algorithms are known to solve 
random fixed-width CNF formulas w.h.p. in polynomial time \cite{Achlioptas09}.
Above a specific clause density the solution space `shatters' into exponentially
many connected-components. Moreover, these clusters are with high probability all linearly separated i.e. the Hamming distance between all pairs of connected-components is
bounded from below by some function linear in the number of variables. This `shattering' of the solution space into linearly separated solutions is known to be difficult for a variety of \SAT-solving algorithms \cite{AM12,Coja11}.

Although this prior work exists on the runtime scaling behavior of \SAT~solvers on random fixed-width CNF formulas and on certain other classes of random constraints, no prior work considers the runtime scaling behavior of \SAT~solvers on formulas composed of both CNF-clauses and XOR-clauses, known as CNF-XOR formulas. Recently, successful hashing-based approaches to the fundamental problems of constrained sampling and counting employ {\SAT} solvers to solve CNF-XOR formulas~\cite{GSS06,CMV13a,ZCSE16,MVCFSFIM16}. The scalability of these hashing-based algorithms crucially depends on the runtime performance of \SAT~solvers in handling CNF-XOR formulas. Although XOR-formulas can be solved individually in polynomial time (using Gaussian Elimination \cite{Sch78}), XOR-formulas are empirically hard \cite{HJKN06} for \SAT~solvers without equivalence reasoning or similar techniques. The rise of applications for CNF-XOR formulas has motivated the development of specialized CNF-XOR solvers, such as \CryptoMiniSAT~\cite{SNC09}, that combine \SAT-solving techniques with algebraic techniques and so can reason about about both the CNF-clauses and XOR-clauses within a single CNF-XOR formula. 

The runtime behavior of these specialized CNF-XOR solvers is an area of active research. Recent work \cite{DMV16} analyzed the satisfiability of random formulas composed of both random $k$-clauses (i.e. CNF-clauses of fixed-width $k$) and random variable-width XOR-clauses (where the width of the XOR-clauses used is stochastic), known as random $k$-CNF-XOR formulas, to begin to demystify the behavior of CNF-XOR solvers. Since the scaling behavior of \SAT~solvers on random $k$-clauses has been analyzed to explain the runtime behavior of \SAT~solvers in practice \cite{Achlioptas09}, we believe that analysis of the scaling behavior of CNF-XOR solvers on random $k$-CNF-XOR formulas is the next step towards explaining the runtime behavior of CNF-XOR solvers in practice and thus explaining the runtime behavior of hashing-based algorithms.

For example, it is widely believed that the performance of CNF-XOR~solvers on CNF-XOR formulas depends on the width of the XOR-clauses. Consequently, recent efforts \cite{Gomes07shortXOR,IMMV16} have focused on designing hashing-based techniques that employ XOR-clauses of smaller width. In this paper, we present empirical evidence that using smaller width XOR-clauses does not necessarily improve the scaling behavior of CNF-XOR~solvers.

The primary contribution of this work is the first empirical and theoretical study of the runtime behavior of CNF-XOR solvers on random $k$-CNF-XOR formulas and on the solution space of random $k$-CNF-XOR formulas. In particular:
\begin{enumerate}
	\item We present (in Section \ref{sec:experiments:scaling}) experimental evidence that the runtime of \CryptoMiniSAT~scales exponentially in the number of variables at many $k$-clause and XOR-clause densities well within the satisfiable region, even when both the CNF and XOR subformulas are separately solvable in polynomial time by \CryptoMiniSAT.
	
	\item We present (in Section \ref{sec:experiments:scaling}) experimental evidence that this exponential scaling peaks around the empirical phase-transition location for random $k$-CNF-XOR formulas, and further that the scaling behavior does \emph{not} monotonically improve as the XOR-clauses get shorter.
	
	\item We prove (in Section \ref{sec:theory}) that the solution space of random variable-width XOR formulas (and therefore of random $k$-CNF-XOR formulas) shatters. We hypothesize that the exponential scaling behavior of random $k$-CNF-XOR formulas within the satisfiable region is caused by this solution space shattering.
\end{enumerate}

%

\section{Notations and Preliminaries} \label{sec:prelims}
Let $X = \{X_1, \cdots, X_n\}$ be a set of propositional variables and 
let $F$ be a formula defined over $X$. A \emph{satisfying assignment} or 
\emph{solution} of $F$ is an assignment of truth values to the variables in $X
$ such that $F$ evaluates to true. The \emph{solution space} of $F$ is the set
of all satisfying assignments. We say that $F$ is \emph{satisfiable} (or \emph{sat.}) if there exists
a satisfying assignment of $F$ and that $F$ is \emph{unsatisfiable} (or \emph{unsat.}) otherwise. 

We describe the solution space of $F$ using terminology from Achlioptas and Molloy \shortcite{AM13}.
Two satisfying assignments $\sigma$ and $\tau$ of $F$ are \emph{$d$-connected}, for a real number $d$, if there exists a sequence of solutions $\sigma, \sigma', \cdots, \tau$ of $F$ such that the Hamming distance of every two successive elements in the sequence is at most $d$. A subset $S$ of the solution space of $F$ is a \emph{$d$-cluster} if every $\sigma, \tau \in S$ is $d$-connected. Two subsets $S$, $S'$ of the solution space of $F$ are \emph{$d$-separated} if every pair $\sigma \in S$ and $\tau \in S'$ is not $d$-connected. Moreover, we say that $F$ is $d$-separated if the Hamming distance between every pair of solutions of $F$ is at least $d$.
 
If $g(n)$ is a function of $n$, we use $O(g(n))$ as shorthand for some function $g'(n) \in O(g(n))$ and use $\Omega(g(n))$ as shorthand for some function $g''(n) \in \Omega(g(n))$ (where the choice of $g'(n)$ and $g''(n)$ is independent of $n$).
 
We use $\P{E}$ to denote the probability of event $E$. 
We say that an infinite sequence of random events $E_1, E_2, \cdots $ occurs 
\emph{with high probability} (denoted, w.h.p.) if $\lim\limits_{n \to \infty} \P{E_n} = 1$. 

A $k$-\emph{clause} (or \emph{CNF-clause}) is the disjunction of $k$ literals out of $\{ X_1, \cdots, X_n\}$, 
with each variable possibly negated. For fixed positive integers $k$ and $n$ and 
a nonnegative real number $r$ (known as the \emph{$k$-clause density}), let the random variable $\F$ 
denote the formula consisting of the conjunction of $\ceil{rn}$ $k$-clauses, each chosen 
uniformly and independently from all $\binom{n}{k}2^k$ possible $k$-clauses over $n$ variables.

The early experiments on $\F$~\cite{MSL92,CA93,KS94} led to the following conjecture:
\begin{conjecture}[Satisfiability Phase-Transition Conjecture]
	For every integer $k \geq 2$, there is a critical ratio $r_k$ such that:
	\begin{enumerate}
		\item If $r < r_k$, then $\F$ is satisfiable w.h.p.
		\item If $r > r_k$, then $\F$ is unsatisfiable w.h.p.
	\end{enumerate}
\end{conjecture}
The Conjecture has been proven for $k=2$ and for all sufficiently large $k$ \cite{DSS15}. The Conjecture has remained elusive for small values 
of $k \geq 3$, although values for these $r_k$ can be estimated experimentally (e.g., $r_3$ seems to be near $4.26$) and predicted analytically using techniques from statistical physics \cite{MMZ06}. 

When the $k$-clause density is small (e.g. below $O(2^k/k)$) there are algorithms that are known to solve $\F$ with high probability in polynomial time \cite{CM97}.
No algorithm is known that can solve $\F$ in polynomial time when the clause density is larger, even when $\F$ is still expected to have exponentially many solutions~\cite{ACR11}. The solution space geometry of $\F$ can be characterized
in the satisfiable region. In particular, for every $k \geq 8$ there exists some $k$-clause density 
$r$ where w.h.p. $\F$ is satisfiable and almost all of the solution space of $\F$ can be partitioned into exponentially many $O(n)$-clusters such that each pair of clusters is $\Omega(n)$-separated \cite{ACR11}. This `shattering' of the solution space into linearly separated clusters is known to be difficult for a variety of \SAT-solving algorithms \cite{AM12,Coja11}.

An XOR-clause over $n$ variables is the `exclusive or' of either 0 or 1 together with a subset 
of the variables $X_1$, $\cdots$, $X_n$. An XOR-clause including 0 (respectively, 1) evaluates to true if and only if an odd (respectively, even) number of the included variables evaluate to true. For a fixed positive integer $n$ and a nonnegative real number $p$, a \emph{random XOR-clause with variable-probability $p$} is an XOR clause $A$ chosen so that each $X_i$ is included in $A$ independently with probability $p$ and $1$ is included in $A$ independently with probability $1/2$. Note that all $k$-clauses contain \emph{exactly} $k$ variables, whereas the number of variables in an XOR-clause is not fixed; a random XOR-clause chosen with variable-probability $p$ over $n$ variables contains $pn$ variables in expectation.
 
For a fixed positive integer $n$, a nonnegative real number $s$ (known as the \emph{XOR-clause density}), and a nonnegative real number $p$ (known as the \emph{XOR variable-probability}), let the random variable $\Q$ denote the formula consisting of the conjunction of $\ceil{sn}$ XOR-clauses, with each clause an independently chosen random XOR-clause with variable-probability $p$. The solution space geometry of $\Q$ has not been characterized in prior work. There is a related model of random XOR-formulas where every XOR-clause contains a fixed number of variables. In this case, w.h.p. the solution space can be partitioned into a set of $O(\log n)$-clusters such that each pair of clusters is $\Omega(n)$-separated \cite{AM13,IKKM12}.

The random variable $\Q[p=1/2]$ matches the XOR-clauses used in several hashing-based constrained sampling and counting algorithms \cite{CMV13a}. Recent work \cite{ZCSE16} has made use of $\Q$ with $p < 1/2$ for constrained sampling and counting algorithms.


A CNF-XOR formula (respectively, $k$-CNF-XOR formula) is the conjunction of some number of CNF-clauses (respectively, $k$-clauses) and XOR-clauses. For fixed positive integers $k$ and $n$ and fixed nonnegative real numbers $r$ and 
$s$, let the random variable $\FQ$ denote the 
formula consisting of the conjunction of $\ceil{rn}$ $k$-clauses, each chosen uniformly and independently from all possible $k$-clauses over $n$ variables, and $\ceil{sn}$ independently chosen XOR-clauses with variable-probability $p$. There exists a phase-transition in the satisfiability of $\FQ[p={1/2}]$ when the $k$-clause density is small, shown by the following theorem \cite{DMV16}: 
\begin{theorem}
	Let $k \geq 2$.	There is a function $\phi_k(r)$ and a constant $\alpha_k \geq 1$ such that for all $s \geq 0$ and all $r \in [0, \alpha_k)$ (except for at most countably many $r$):
	\begin{enumerate}
		\item If $s < \phi_k(r)$, then w.h.p. $\FQ[p={1/2}]$ is sat.
		\item If $s > \phi_k(r)$, then w.h.p. $\FQ[p={1/2}]$ is unsat.
	\end{enumerate}
\end{theorem}

%
\section{Experimental Results} \label{sec:experiments:scaling}
To explore empirically the runtime behavior of solvers on randomly constructed $k$-CNF-XOR formulas, we built a prototype implementation in Python that employs the {\CryptoMiniSAT}\footnote{\url{http://www.msoos.org/cryptominisat4/} }~\cite{SNC09} solver to check satisfiability of random $k$-CNF-XOR formulas. We chose {\CryptoMiniSAT} because it is typically used in hashing-based approaches to sampling and counting due to its ability to handle the combination of $k$-clauses and XOR-clauses efficiently~\cite{CFMSV14}. 

The objective of the experimental setup was to empirically determine the scaling behavior, as a function of $n$, in the median runtime of checking satisfiability of $\FQ$ with respect to $r$ (the $k$-clause density), $s$ (the XOR-clause density), and $p$ (the XOR variable-probability) for fixed $k$.

\subsection{Experimental Setup} \label{sec:experiments:scaling:setup}
To uniformly choose a $k$-clause we uniformly selected without replacement 
$k$~out of the variables $\{X_1, \cdots, X_n\}$. For each selected variable $X_i$, 
we include exactly one of the literals $X_i$ or $\neg X_i$ in the $k$-clause, 
each with probability $1/2$. The disjunction of these $k$ literals is 
a uniformly chosen $k$-clause. To choose an XOR-clause with variable-probability $p$, we include 
each variable of $\{X_1, \cdots, X_n\}$ with probability $p$ in a set 
$A$ of variables. We also include in $A$ exactly one of $0$ or $1$, 
each with probability $1/2$. The `exclusive-or' of all elements of $A$ 
is a random XOR-clause with variable-probability $p$.

In all experiments we fix the clause length $k=3$. The $3$-clause density $r$, the XOR-clause density $s$, and the XOR variable-probability $p$ varied in each experiment, as follows:
\begin{itemize}
	\item To study the effect of the $3$-clause and XOR-clause densities on the runtime, we ran 124 experiments with $r \in \{1, 2, 3, 4\}$, $p = 1/2$, and $s$ ranging from 0.3 to 0.9 in increments of 0.02. We present results from these experiments in Section \ref{sec:experiments:scaling:results}.
	
	\item To study the effect of the XOR variable-probability on the runtime, we ran 1295 experiments with $r = 2$,
	$p$ ranging from 0.02 to 0.94 in increments of 0.005, and $s$ ranging from 0.3 to 0.9 in increments of 0.1. We chose these clause-densities so that approximately half of the clause-densities were in the satisfiable region. We present selected results from these experiments in Section \ref{sec:experiments:scaling:alternative}.
\end{itemize}

To determine the scaling behavior of \CryptoMiniSAT~on random $k$-CNF-XOR formulas with parameters $k$, $r$, $s$, and $p$, we determined a number of variables $N$ so that the median runtime of \CryptoMiniSAT~on $\FQ[n={N}]$ was as large as possible while remaining below the set formula timeout. We then allowed $n$ to range from $10$ to $N$ in increments of $1$.
For each $n$, we used {\CryptoMiniSAT} to check the satisfiability of $100$ formulas sampled from $\FQ$ by constructing the conjunction of $\ceil{rn}$ $k$-clauses and $\ceil{sn}$ XOR-clauses, with each clause chosen independently as described above. The solving of each formula was individually timed. The median runtime is an estimate for the median runtime of \CryptoMiniSAT~on $\FQ$.

Finally, we used the curve\_fit function in the Python scipy.optimize\footnote{\url{https://www.scipy.org/}} library to determine the relationship between the number of variables $n$ and the medium runtime of {\CryptoMiniSAT}~on $\FQ$. We attempted to fit linear ($an+b$), quadratic ($an^2 +bn + c$), cubic ($an^3 + bn^2 + cn + d$), and exponential ($\beta 2^{\alpha n}$) curves; the best-fit curve was the curve with the smallest mean squared error.




Each experiment was run on a node within a high-performance computer cluster. 
These nodes contain 12-processor cores at 2.83 GHz each with 48 GB of RAM per node. 
Each formula was given a timeout of 10 seconds. We were not able to run informative experiments for formulas with higher timeouts; as the runtime of \CryptoMiniSAT~increases past 10 seconds, the variance in runtime significantly increases as well and so experiments require a number of trials at each data point far beyond our computational abilities.

\begin{figure}
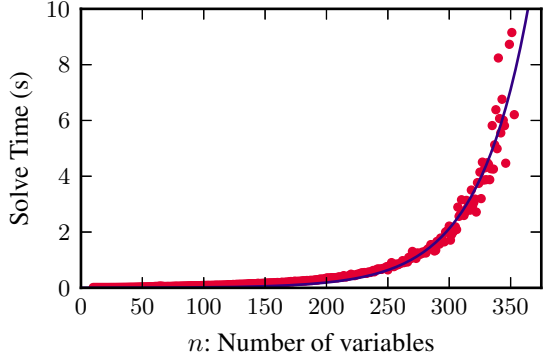

	\centering
\begingroup%
\makeatletter%
%
\makeatother%
\endgroup%


	\vspace{0in}
	\caption{Runtime for $3$-CNF-XOR formulas at $3$-clause density $r=2$, XOR-clause density $s=0.3$, and XOR variable-probability $p=1/2$, together with the best-fit curve $0.00152 \cdot 2^{0.0348n}$. \label{fig:scaling_r4_s03}}
\end{figure}

\begin{figure}
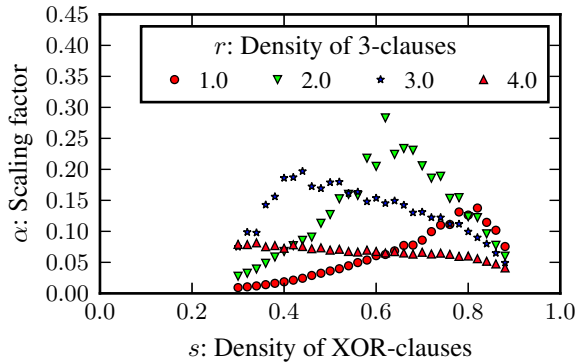

	\centering
\begingroup%
\makeatletter%
%
\makeatother%
\endgroup%


	\vspace{0in}
	\caption{Exponential scaling factor for $3$-CNF-XOR formulas with $3$-clause density $r\in\{1, 2, 3, 4\}$ and XOR variable-probability $p=1/2$. The scaling factor $\alpha$ is the exponent of the best-fit line for the runtime of $\FQ[k=3, p={1/2}]$.  \label{fig:scaling_r23_scalefactor_p5}}
\end{figure}

\subsection{Results on the Impact of XOR-clause Density} \label{sec:experiments:scaling:results}
We analyzed the median runtime of {\CryptoMiniSAT}~on $\FQ[k=3, p=1/2]$ for a fixed $r$ and $s$ as a function of the number of variables $n$.

Figure \ref{fig:scaling_r4_s03} plots the median runtime at $k=3$, $r=2$, and $s=0.3$ as a function of $n$, together with the best-fit curve. The x-axis indicates the number of variables $n$. The y-axis indicates the median runtime of \CryptoMiniSAT on $\FQ[k=3, r=4, s=0.3, p=1/2]$. We observe that
the median runtime increases exponentially in the number of variables. In this case, the best-fit curve is the exponential function $0.00152 \cdot 2^{0.0348n}$.

In fact, for all experiments with $r \in \{1, 2, 3, 4\}$ and $0.3 \leq s \leq 0.9$ the best-fit curve to the median runtime as a function of $n$ is proportional to an exponential function of the form $2^{\alpha n}$ for some $\alpha > 0$. Figure \ref{fig:scaling_r23_scalefactor_p5} plots the scaling behavior with respect to $n$ of 
the median runtime of {\CryptoMiniSAT}~on $\FQ[p=1/2, k=3]$. The x-axis indicates the density of XOR-clauses $s$. The legend indicates the density of $3$-clauses $r$. The value $\alpha$, known as the \emph{scaling factor}, shown on the y-axis indicates that the best-fit curve to the median runtime of $\FQ[p=1/2,k=3]$ as a function of $n$ was proportional to $2^{\alpha n}$. We observe that the scaling factor is closely related to the $3$-clause density and the XOR-clause density: when the XOR-clause density is low or high the scaling factor is low, and the scaling factor peaks at some intermediate value. We observe peaks in the scaling factor near ($r=1$, $s=0.8$), ($r=2$, $s=0.6$) and ($r=3$, $s=0.4$). Empirically, there is a phase-transition in the satisfiability of random $3$-CNF-XOR formulas exactly at these locations \cite{DMV16}. Thus we observe a peak in the runtime scaling factor around the $3$-CNF-XOR phase-transition, similar to the peak observed in the runtime factor for $\F$ around the $k$-CNF phase-transition \cite{CDSSV03}.

Our experimental results do not describe extremely low $3$-clause densities and XOR-clause densities (for example, when the XOR-clause density is below $0.3$). At such low densities, conclusive evidence of polynomial or exponential behavior requires computational power beyond our capabilities.

%
\subsection{Results on the Impact of XOR-clause Width} \label{sec:experiments:scaling:alternative}
We next analyzed the median runtime of {\CryptoMiniSAT}~on $\FQ[k=3, r=2]$ for a fixed $p$ and $s$ as a function of the number of variables $n$. For lack of space, we present results only for the experiments with $s \in \{0.4, 0.7\}$ \footnote{
	The data from all experiments is available at \url{http://www.cs.rice.
	edu/CS/Verification/Projects/CUSP/}
}.

\begin{figure}
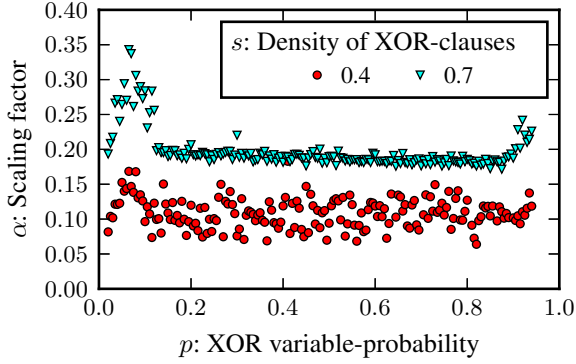

	\centering
\begingroup%
\makeatletter%
%
\makeatother%
\endgroup%


	\vspace{0in}
	\caption{Exponential scaling factor for $3$-CNF-XOR formulas with $3$-clause density $r=2$ and XOR-clause densities $s=0.4~\text{and}~0.7$. \label{fig:scaling_r2_scalefactor_p}}
\end{figure}

Figure \ref{fig:scaling_r2_scalefactor_p} plots the scaling behavior with respect to $n$ of the medium runtime of {\CryptoMiniSAT}~on $\FQ[k=3, r=2]$. The x-axis indicates the XOR variable-probability $p$. The legend indicates the density of XOR-clauses $s$. The value $\alpha$ shown on the y-axis indicates that the best-fit curve to the median runtime of $\FQ[k=3, r=2]$ as a function of $n$ was proportional to $2^{\alpha n}$. Note that $(r=2, s=0.4)$ is in the satisfiable region and $(r=2, s=0.7)$ is in the unsatisfiable region when $p=1/2$. 
We observe that the behavior of the scaling factor is independent of the XOR variable-probability, $p$, when $p \in (0.15, 0.9)$. As $p$ decreases below $0.15$, the scaling factor increases to a peak when $p \in (0.05, 0.1)$, then decreases. We also observe a peak in the scaling factor when $p > 0.9$.

In summary, we observe that the runtime of \CryptoMiniSAT~scales exponentially in the number of variables on random $3$-CNF-XOR formulas across a wide range of densities and XOR variable-probabilities. The exponential scaling behavior peaks near the empirical location of the $3$-CNF-XOR phase-transition. The exponential scaling behavior is constant when the XOR variable-probability is between $0.15$ and $0.9$ and the scaling behavior peaks when the XOR variable-probability is between $0.05$ and $0.1$, independent of the XOR-clause density.

\section{The Separation of the XOR Solution Space}
\label{sec:theory}
In the case of $k$-CNF formulas, the exponential runtime scaling of DPLL-solvers (in the satisfiable region) is closely connected to the `shattering' of the solution space into exponentially many $\Omega(n)$-separated clusters w.h.p.~\cite{AM12,Coja11}. Prior work has shown that the solution space of fixed-width XOR-clauses has similar behavior; unfortunately, the proof techniques used in this prior work do not easily extend to the solution space of $\Q$. In particular, the proof techniques for XOR-clauses of fixed-width $\ell$ heavily involve properties of either random $\ell$-uniform hypergraphs \cite{AM13} or random factor graphs with factors of constant degree $\ell$ \cite{IKKM12}. If the width of each XOR-clause is stochastic, as in $\Q$, rather than fixed, the corresponding hypergraphs are not uniform and the corresponding factor graphs do not have factors of constant degree.

Nevertheless, we show in Theorem \ref{thm:separation} that all solutions of a random XOR-formula are w.h.p. $\Omega(n)$-separated (as long as the variable-probability decreases slowly enough as a function of $n$). This is a stronger separation than the separation seen in the case of $k$-CNF formulas and fixed-width XOR-formulas, where there may be clusters of nearby solutions.
\begin{theorem}
	\label{thm:separation}
	Let $s \in (0, 1)$, $\rho > 2$, and $f(n)$ be a nonnegative function. If $\rho \frac{\log(sn)}{sn} \leq f(n) \leq 1/2$ for all large enough $n$, then $\Q[p=f(n)]$ is w.h.p. $\Omega(n)$-separated.
\end{theorem}
\begin{proof}
	This follows directly from Lemma \ref{lemma:separation:logn_to_separation}. The proof of this lemma appears in Section \ref{sec:theory:proofs}.
\end{proof}

Notice that Theorem \ref{thm:separation} allows the XOR variable-probability to depend on the number of variables. In particular, the XOR variable-probability can decrease as a function of $n$. Theorem \ref{thm:separation} does not characterize the solution space of XOR-formulas when the variable-probability decreases faster than $2 \frac{\log(sn)}{sn}$ as a function of $n$. It is possible that the solution space is still $\Omega(n)$-separated in this case, or that clusters of solutions can be found. We leave this for future work.

In Section \ref{sec:experiments:scaling}, we focused on an XOR variable-probability model that is independent of $n$; this XOR variable-probability is an important special case of the above general theorem. In particular, if the XOR variable-probability is some constant $p \in (0, 1/2]$ then the solution space of a random XOR-formula with variable-probability $p$ is $\Omega(n)$-separated. We highlight this fact as Corollary \ref{thm:separation:cases}.

\renewcommand{\theenumi}{(\alph{enumi})}
\begin{corollary}
	\label{thm:separation:cases}
	For all $s \in (0, 1)$ and $p \in (0, 1/2]$, $\Q$ is w.h.p. $\Omega(n)$-separated.
\end{corollary}
\begin{proof}
	This follows from Theorem \ref{thm:separation} with $f(n) = p$.
\end{proof} 

Corollary \ref{thm:separation:cases} also implies that $\FQ = \F\land \Q$ is w.h.p. $\Omega(n)$-separated. Since the separation of the $k$-CNF solution space is closely connected to the exponential scaling of \SAT~solvers, we hypothesize that the exponential scaling of \CryptoMiniSAT~we observed in Section \ref{sec:experiments:scaling} at many XOR-clause densities and XOR variable-probabilities is closely connected to the $\Omega(n)$-separation of $k$-CNF-XOR formulas shown in Corollary \ref{thm:separation:cases} at all nonzero XOR-clause densities and XOR variable-probabilities (below $1/2$).

\subsection{Proofs}
\label{sec:theory:proofs}
In this section we establish Theorem \ref{thm:separation}, which follows directly from Lemma \ref{lemma:separation:logn_to_separation}.
To do this, notice that if two solutions of $\Q$ differ exactly on a set of variables $A$ then every XOR-clause in $\Q$ must contain an even number of variables from $A$. We bound from above the probability that a random XOR-clause chosen with variable-probability $p$ contains an even number of variables from $A$. By summing this bound across all sets $A$ containing no more than $\lambda n$ variables for some constant $\lambda$, we bound the probability that two solutions to $\Q$ differ in no more than $\lambda n$ variables.

The following lemma presents an elementary result in probability theory. We use this result in Lemma \ref{lemma:sum_to_separation} to bound the probability that a random XOR-clause chosen with variable-probability $p$ has an even number of variables from a set $A$.
\begin{lemma}\label{lemma:bernoulli}
Let $N$ be a positive integer and let $p$ be a real number with $0 \leq p \leq 1$. If $B_1, B_2, \cdots, B_N$ are independent Bernoulli random variables with parameter $p$, then $\P{\sum_{i=1}^N B_i~\text{is even}} = 1/2 + 1/2(1 - 2p)^N$.
\end{lemma}  
\begin{proof}
	Fix $p \in [0, 1]$. For all $N \geq 0$, let $a_N$ be the probability that the sum of $n$ independent Bernoulli random variables with parameter $p$ is even. Then $a_0 = 1$ and $a_N = (1-p)a_{N-1} + p(1 - a_{N-1}) = p + a_{N-1} - 2pa_{N-1}$ for all $N \geq 1$. It follows that $a_N = 1/2 + 1/2 (1-2p)^N$.
\end{proof}

The following lemma shows that the sum of these probabilities across all sets whose size is smaller than $\lambda n$ goes to 0 in the limit as $n \to \infty$ when the XOR variable-probability is proportional to $\log(sn)/(sn)$.

\begin{lemma}
	\label{lemma:sum:logn}
	Let $\alpha, \delta \in (0, 1)$, $m = \alpha n$, $\kappa > -\frac{\log(2/(1+\delta) - 1)}{\log(1+\delta)}$ and $\lambda^* < 1/2$ such that $-\lambda^* \log(\lambda^*) - (1 - \lambda^*)\log(1 - \lambda^*) = \alpha \log(1 + \delta)$. Then for all $\lambda < \lambda^*$:
	$$\lim_{n \to \infty} \sum_{w=1}^{\lambda n} {{n} \choose {w}} \left(\frac{1}{2} + \frac{1}{2}\left(1 - 2\kappa \frac{\log m}{m}\right)^w\right)^m = 0$$
\end{lemma}
\begin{proof}		
	This proof is given as \textbf{Lemma 7} of \cite{ZCSE16}.		
\end{proof}

The following lemma allows us to show that the XOR solution-space is $g(n)$-separated if the XOR variable-probability is $f(n)$ for some functions $f$ and $g$ provided that the sum of probability of all sets of variables whose size is below $g(n)$ goes to 0. In particular, in Lemma \ref{lemma:separation:logn_to_separation} we use this lemma with $f(n) \propto \log(sn)/(sn)$ and $g(n) \in \Omega(n)$ to show that the solution-space of $\Q[p=f(n)]$ is $\Omega(n)$-separated.
\begin{lemma}
	\label{lemma:sum_to_separation}
	Let $f(n)$ and $g(n)$ be nonnegative functions with $f(n) \leq 1$ for all sufficiently large $n$. If $$\lim_{n \to \infty} \sum_{w=1}^{g(n)} {{n} \choose {w}} \left(\frac{1}{2} + \frac{1}{2}\left(1-2f(n)\right)^w\right)^{sn} = 0$$ then w.h.p. all solutions of $\Q[p={f(n)}]$ are $g(n)$-separated.	
\end{lemma}
\begin{proof}
	Let the random variable $D$ be 1 if $\Q[p={f(n)}]$ has two solutions with a Hamming distance less than or equal to $g(n)$ and 0 otherwise. We would like to prove that $\lim_{n \to \infty} \P{D = 1} = 0$.
		
	For all nonempty subsets of variables $A \subseteq X$, let the random variable $D(A)$ be 1 if $\Q$ has a pair of solutions that differ exactly on the variables of $A$ and 0 otherwise. Then $D(A) = 1$ if and only if each XOR-clause in $Q$ contains an even number of variables from $A$. Moreover, let $\mathcal{B}$ be the set of all subsets of variables $A \subseteq X$ s.t. $0 < |A| \leq g(n)$ and notice that $D \leq \sum_{A\in \mathcal{B}} D(A)$. Thus $\P{D=1} \leq \sum_{A \in \mathcal{B}} \P{D(A) = 1}$.
		
	Fix $A \subseteq X$ and let $Q_1$ be a random XOR-clause chosen with variable-probability $f(n)$.
	Enumerate the $|A|$ variables in $A$ as $Y_1, Y_2, \cdots, Y_{|A|}$. Then for all $1 \leq i \leq |A|$ we define a random variable $B_i$ that is $1$ if the variable $Y_i$ appears in $Q_1$ and is $0$ otherwise. Notice that each $B_i$ is an independent Bernoulli random variable with parameter $f(n)$, and further that the number of variables from $A$ contained in $Q_1$ is exactly $\sum_{i=1}^{|A|} B_i$. By Lemma \ref{lemma:bernoulli} it follows that the probability that $Q_1$ contains an even number of variables from $A$ is $1/2 + 1/2(1-2f(n))^{|A|}$. 
	
	Since all $\ceil{sn}$ XOR-clauses of $\Q$ are chosen independently with variable-probability $f(n)$, it follows that $\P{D(A)} = (1/2 + (1-2f(n))^{|A|}/2)^{\ceil{sn}}$. For all sufficiently large $n$, we have $0 \leq f(n) \leq 1$ and thus $0 \leq 1/2 + (1-2f(n))^{|A|}/2 \leq 1$. Thus $\P{D(A)} \leq (1/2 + (1-2f(n))^{|A|}/2)^{sn}$ for all sufficiently large $n$.
	
	Finally, notice that there are exactly ${{n} \choose {w}}$ sets in $\mathcal{B}$ of size $w \leq g(n)$ and so $\P{D=1} \leq \sum_{A \in \mathcal{B}} (1/2 + (1-2f(n))^{|A|}/2)^{sn} = \sum_{w=1}^{g(n)} {{n} \choose {w}} (1/2 + (1-2f(n))^{w}/2)^{sn}$. By hypothesis, this implies that $\lim_{n \to \infty} \P{D = 1} = 0$.
\end{proof}

The following lemma combines Lemma \ref{lemma:sum:logn} and Lemma \ref{lemma:sum_to_separation} to show that a variable-probability above $2\log(sn)/(sn)$ implies $\Omega(n)$-separation. This finishes the proof of Theorem \ref{thm:separation}.
\begin{lemma}
	\label{lemma:separation:logn_to_separation}
	Let $s$ and $\rho$ be real numbers such that $0 < s \leq 1$ and $\rho > 2$. If $f(n)$ is a nonnegative function such that $\rho \frac{\log(sn)}{sn} \leq f(n) \leq 1/2$ for all sufficiently large $n$, then $\Q[p={f(n)}]$ is w.h.p. $\Omega(n)$-separated.
\end{lemma}
\begin{proof}
	Let $a(x) = -\log(2/(1+x) - 1) / \log(1+x)$. Notice that $\lim_{x \to 0} a(x) = 2$, $\lim_{x \to 1} a(x) = \infty$, and $a(x)$ is continuous on $(0, 1)$. Since $2 < \rho < \infty$, it follows that there is some $\delta \in (0, 1)$ with $a(\delta) < \rho$.
	
	Let $H(x) = -x\log(x) - (1-x) \log(1-x)$. Since $H$ is continuous on $[0, 1/2]$ and $H(0) < s \log(1 + \delta) < H(1/2)$, it follows there is some $\lambda^* \in (0, 1/2)$ with $H(\lambda^*) = s \log(1+\delta)$. Define $\widehat{f}(n) = \rho \frac{\log(sn)}{sn}$ and $g(n) = n \lambda^*/2$. 
	
	Then by Lemma \ref{lemma:sum:logn} with $\alpha = s$, $\kappa = \rho$, and $\lambda = \lambda^*/2$ (and with $\delta$ and $\lambda^*$ as defined above) we have that $\lim_{n \to \infty} \sum_{w=1}^{g(n)} {{n} \choose {w}} (1/2 + (1-2\widehat{f}(n))^w/2)^{sn} = 0$. 
	
	Notice that, for all sufficiently large $n$, $\widehat{f}(n) \leq f(n) \leq 1/2$ and so $(1-2\widehat{f}(n))^w \geq (1-2f(n))^w \geq 0$ for all $w \geq 1$. Therefore ${{n} \choose {w}} (1/2 + (1-2\widehat{f}(n))^w/2)^{sn} \geq {{n} \choose {w}} (1/2 + (1-2f(n))^w/2)^{sn}$ for all $w \geq 1$ and for all sufficiently large $n$. Thus $\lim_{n \to \infty} \sum_{w=1}^{g(n)} {{n} \choose {w}} (1/2 + (1-2f(n))^w/2)^{sn} = 0$ and so by Lemma \ref{lemma:sum_to_separation} we conclude that $\Q[p={f(n)}]$ is $\Omega(g(n)) = \Omega(n \lambda^*/2) = \Omega(n)$-separated w.h.p. as desired.	
\end{proof}


\section{Conclusion}\label{sec:conclusion}
We presented the first study of the runtime behavior of \SAT~solvers on random $k$-CNF-XOR formulas. We presented experimental evidence that \CryptoMiniSAT~scales exponentially on random $k$-CNF-XOR formulas across a wide range of $k$-clause densities, XOR-clause densities, and XOR variable-probabilities. To begin to explain this phenomenon in the satisfiable region, we proved that the solution space of XOR-formulas is linearly separated w.h.p..

Recent hashing-based algorithms for sampling and counting allow some freedom in the exact parameters (for example, in the XOR-clause density \cite{CMV16} or the XOR variable-probability \cite{ZCSE16}) used to generate CNF-XOR formulas. This paper suggests combinations of clause-densities and XOR variable-probabilities that are likely to be difficult for CNF-XOR solvers and thus should be avoided. Using this information to develop better heuristics for hashing-based algorithms is an exciting direction for future work that may lead to significant runtime improvements. For a more detailed discussion, see \cite{Dudek17}.


\section*{Acknowledgments}
Work supported in part by NSF grants CCF-1319459 and IIS-1527668, by NSF Expeditions in Computing project ``ExCAPE: Expeditions in Computer Augmented Program Engineering", by BSF grant 9800096,  by the Ken Kennedy Institute Computer Science \& Engineering
 Enhancement Fellowship funded by the Rice Oil \& Gas HPC Conference, and Data Analysis and 
 Visualization Cyberinfrastructure funded by NSF under grant 
 OCI-0959097. Kuldeep S. Meel is supported by the IBM PhD Fellowship and the Lodieska Stockbridge Vaughn Fellowship.


{\small 

}

\end{document}